\documentclass[11pt]{article}

\usepackage{graphicx,amssymb,amsmath}
\usepackage{amsfonts}
\usepackage{booktabs}
\usepackage{cite}
\usepackage{color}
\usepackage{enumerate}
\usepackage{float}
\usepackage{hyperref}
\usepackage{mathrsfs}
\usepackage{subfig}
\usepackage{tabularx}
\usepackage{authblk}

\usepackage[in]{fullpage}
\usepackage[top=0.8in, bottom=0.9in, left=0.9in, right=0.9in]{geometry}
\def\calA{\mathcal{A}}

\newtheorem{lemma}{Lemma}
\newtheorem{theorem}{Theorem}
\newenvironment{proof}{\noindent {\textbf{Proof:}}\rm}{\hfill $\Box$ \rm\bigskip}

\title{A Simple Algorithm for Computing the Zone of a Line in an Arrangement of Lines\thanks{This research was supported in part by NSF under Grant CCF-2005323.}
}

\author{
Haitao Wang
}
\affil{Department of Computer Science \\
Utah State University, Logan, UT 84322, USA \\ {\tt haitao.wang@usu.edu}}

\begin{document}

\pagestyle{plain}
\pagenumbering{arabic}
\setcounter{page}{1}
\date{}
\maketitle

\vspace{-0.35in}
\begin{abstract}
Let $L$ be a set of $n$ lines in the plane. The zone $Z(\ell)$ of a line $\ell$ in the arrangement $\mathcal{A}(L)$ of $L$ is the set of faces of $\mathcal{A}(L)$ whose closure intersects $\ell$. It is known that the combinatorial size of $Z(\ell)$ is $O(n)$. Given $L$ and $\ell$, computing $Z(\ell)$ is a fundamental problem.
Linear-time algorithms exist for computing $Z(\ell)$ if $\mathcal{A}(L)$ has already been built, but building $\mathcal{A}(L)$ takes $O(n^2)$ time. On the other hand, $O(n\log n)$-time algorithms are also known for computing $Z(\ell)$ without relying on $\mathcal{A}(L)$, but these algorithms are relatively complicated.
In this paper, we present a simple algorithm that can compute $Z(\ell)$ in $O(n\log n)$ time. More specifically, once the sorted list of the intersections between $\ell$ and the lines of $L$ is known, the algorithm runs in $O(n)$ time.
A big advantage of our algorithm, which mainly involves a Graham's scan style procedure, is its simplicity.
\end{abstract}


\section{Introduction}
\label{sec:intro}
Given a set $L$ of $n$ lines in the plane, let $\calA(L)$ denote the arrangement of the lines of $L$, i.e., the subdivision of the plane induced by $L$. For a line $\ell$, the {\em zone} of $\ell$ in the arrangement $\calA(L)$ is the set of faces of $\calA(L)$ whose closure intersects $\ell$ (see Fig.~\ref{fig:zonefig}); we use $Z(\ell)$ to denote the zone of $\ell$. Given $L$ and $\ell$, we consider the problem of constructing $Z(\ell)$.

It has been proved that the combinatorial size of the zone $Z(\ell)$ is bounded by $O(n)$~\cite{ref:ChazelleTh85,ref:EdelsbrunnerCo86,ref:EdelsbrunnerOn93,ref:EdelsbrunnerAr92,ref:EdelsbrunnerAl87,ref:BernHo91}.
The problem of computing $Z(\ell)$ is a fundamental problem in computational geometry. If the arrangement $\calA(L)$ has already been explicitly constructed, then $Z(\ell)$ can be computed in $O(n)$ time~\cite{ref:ChazelleTh85,ref:EdelsbrunnerCo86}. Indeed, this leads to an incremental algorithm for constructing the arrangement $\calA(L)$ in $O(n^2)$ time~\cite{ref:ChazelleTh85,ref:EdelsbrunnerCo86}. Without having the arrangement $\calA(L)$, computing $Z(\ell)$ can be done in $O(n\log n)$ time. For example, Alevizos, Boissonnat, and Preparata~\cite{ref:AlevizosAn90} proved that the size of any cell in an arrangement of a set of $n$ rays in the plane is $O(n)$ and gave an $O(n\log n)$ time algorithm to construct any cell. Their algorithm can be used to compute $Z(\ell)$ in $O(n\log n)$ time. Indeed, we can cut each line of $L$ into two rays at its intersection with $\ell$. Then, for each side of $\ell$, we compute the cell containing $\ell$ in the arrangement of the rays on that side of $\ell$. The zone $Z(\ell)$ can be obtained from the two cells computed above.

In this paper, we present a new algorithm for computing the zone $Z(\ell)$ in $O(n\log n)$ time. More specifically, once the sorted list of all intersections between $\ell$ and the lines of $L$ is known, the algorithm runs in $O(n)$ time. In contrast, even if the above sorted list is known, applying the algorithm of~\cite{ref:AlevizosAn90} to compute $Z(\ell)$ still takes $O(n\log n)$ time because the algorithm involves sweeping line procedures that are modifications of the classical algorithm for computing the intersections of line segments. A big advantage of our algorithm is that it is quite simple. Indeed, a main process of our algorithm is a Graham's scan style procedure, which is a textbook level algorithm.
As computing $Z(\ell)$ is a fundamental problem and many algorithms use it as a subroutine (e.g.,~\cite{ref:AgarwalPa902,ref:AgarwalCo98,ref:EdelsbrunnerIm89}), it is worth pursuing a simple algorithm.

Many other problems in arrangements of lines or other curves are also fundamental and have been extensively studied. We refer the reader to \cite{ref:AgarwalAr00,ref:EdelsbrunnerAl87,ref:HalperinAr17,ref:SharirDa95} for some excellent books and surveys.

\begin{figure}[t]
\begin{minipage}[t]{0.49\textwidth}
\begin{center}
\includegraphics[height=1.8in]{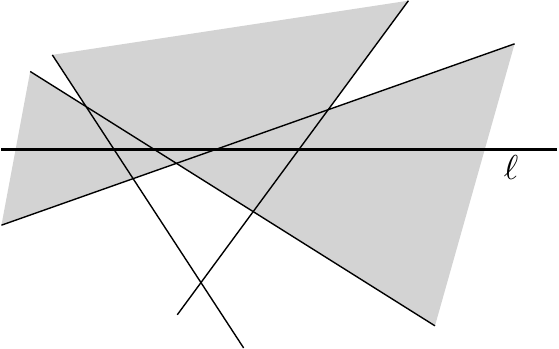}
\caption{\footnotesize The shaded region is the zone $Z(\ell)$.}
\label{fig:zonefig}
\end{center}
\end{minipage}
\hspace{-0.02in}
\begin{minipage}[t]{0.49\textwidth}
\begin{center}
\includegraphics[height=1.8in]{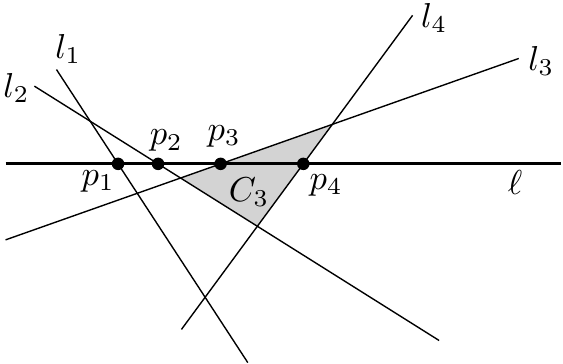}
\caption{\footnotesize Illustrating the points $p_i$ and a cell $C_3$.}
\label{fig:cell}
\end{center}
\end{minipage}
\vspace{-0.15in}
\end{figure}

\section{The algorithm}

Without loss of generality, we assume that $\ell$ is the $x$-axis. To simplify the discussion, we make a general position assumption that no line of $L$ is horizontal and no two lines of $L$ have an intersection on $\ell$. We will discuss at the end of this section that degenerate cases can be easily handled, although  standard techniques~\cite{ref:EdelsbrunnerTo89,ref:EdelsbrunnerSi90} could be applied too.

A curve $\gamma$ in the plane is {\em $y$-monotone} if any horizontal line either does not intersect $\gamma$ or it intersects $\gamma$ at a single point. We say that a $y$-monotone convex curve $\gamma$ is {\em rightward} (resp. {\em leftward}) if $\gamma$ always makes right turns from its lower endpoint to its upper endpoint.

Due to the general position assumption, every line of $L$ intersects $\ell$ at a point. We start by computing the intersections between $\ell$ and all lines of $L$, and then sort them. This takes $O(n\log n)$ time. The rest of the algorithm runs in $O(n)$ time.
Let $Z^+(\ell)$ denote the portion of the zone $Z(\ell)$ above $\ell$ and $Z^-(\ell)$ the portion of $Z(\ell)$ below $\ell$. In the following, we describe an algorithm to compute $Z^+(\ell)$ in $O(n)$ time; $Z^-(\ell)$ can be computed in $O(n)$ time analogously.

Let $l_1,l_2,\ldots,l_n$ be the sorted list of the lines of $L$ from left to right by their intersections with $\ell$. Due to our general position assumption, this order is unique.
For each $1\leq i\leq n$, define $p_i$ as the intersection of $l_i$ and $\ell$ (see Fig.~\ref{fig:cell}). The point $p_i$ divides $l_i$ into two half-lines, and we use $l_i^+$ to refer to the one above $\ell$.
Hence, $p_i$ is the lower endpoint of $l_i^+$; for reference purpose, we also assume that $l_i^+$ has an upper endpoint at infinity and use $p_i'$ to denote it.
For convenience, let $p_0$ represent the left endpoint of $\ell$ at $-\infty$ and $p_{n+1}$ the right endpoint of $\ell$ at $\infty$.

It is easy to see that for each $0\leq i\leq n$, the segment $\overline{p_ip_{i+1}}$ is contained in a single cell of $\calA(L)$, denoted by $C_i$, which is also a cell in $Z(\ell)$ (see Fig.~\ref{fig:cell}). Denote by $C_i^+$ the portion of $C_i$ above $\ell$. Observe that $Z^+(\ell)$ is the disjoint union of cells $C_i^+$ for all $i=0,1,\ldots,n$.
Hence, it suffices to compute the cells $C_i^+$ for all $i=0,1,\ldots,n$.

\begin{figure}[h]
\begin{minipage}[t]{\textwidth}
\begin{center}
\includegraphics[height=2.5in]{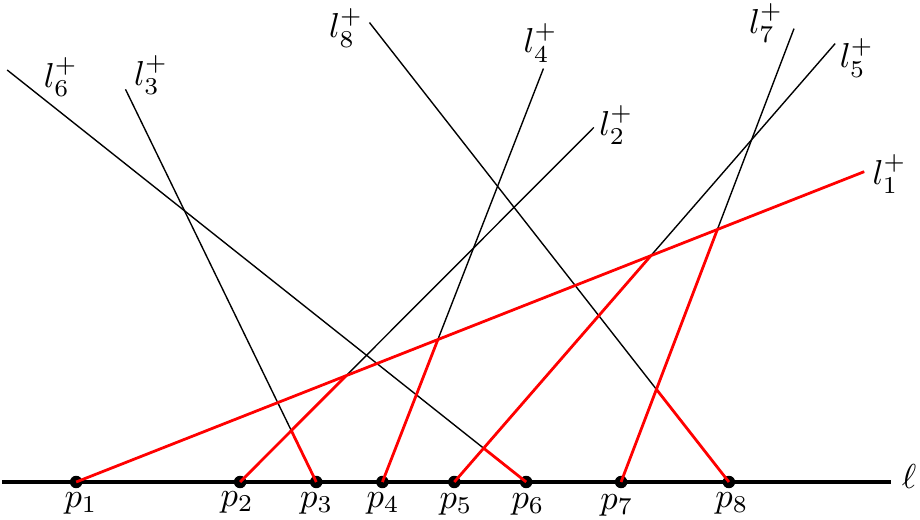}
\caption{\footnotesize The forward forest $F$ is colored red.}
\label{fig:tree}
\end{center}
\end{minipage}
\end{figure}

\begin{figure}[h]
\begin{minipage}[t]{\textwidth}
\begin{center}
\includegraphics[height=2.5in]{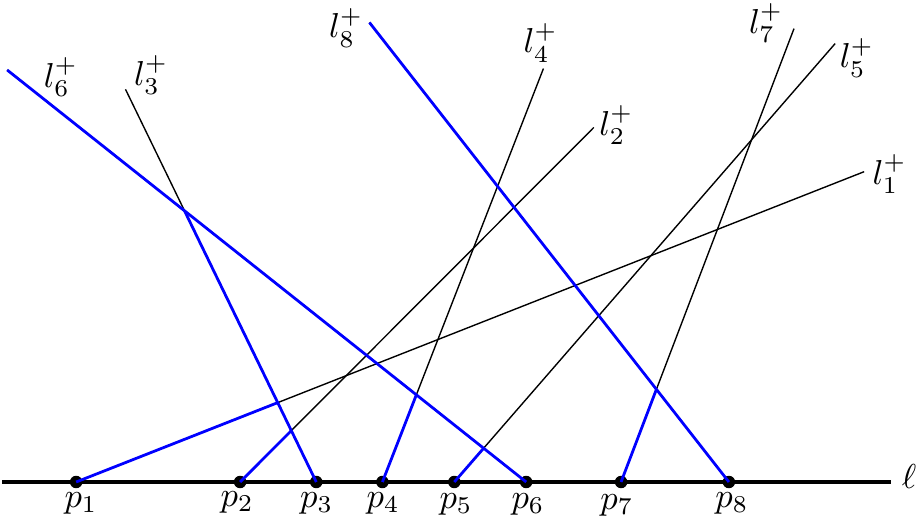}
\caption{\footnotesize The backward forest $F'$ is colored blue.}
\label{fig:backtree}
\end{center}
\end{minipage}
\vspace{-0.15in}
\end{figure}

\paragraph{Forward and backward forests.}
With respect to the index order $1,2,\ldots,n$, we define a {\em forward forest} $F$ as follows (see Fig.~\ref{fig:tree}). Let $F_1=l_1^+$. For each $2\leq i\leq n$, $F_i$ is obtained from $F_{i-1}$ by adding to $F_{i-1}$ the segment of $l^+_i$ from $p_i$ to its first intersection with $F_{i-1}$ (if $l_i^+$ does not intersect $F_{i-1}$, then the entire $l^+_i$ is included in $F_i$). Let $F=F_n$. It is not difficult to see that $F$ has at most $2n-1$ edges because $F_i$ has at most two more edges than $F_{i-1}$.
We can view $F$ as a forest in which the leaves are the points $p_i$ for all $1\leq i\leq n$ and the roots are upper endpoints of some half-lines $l_i^+$ (e.g., in Fig.~\ref{fig:tree}, $F$ consists of only one tree, whose root is the upper endpoint of $l_1^+$). That is why we call $F$ a forest. Notice that if we move from a point $p_i$ along $F$ until the root of the tree containing $p_i$, we always turn rightward, i.e., the path from $p_i$ to the root is a $y$-monotone rightward convex chain.

Similarly, we define a {\em backward forest} $F'$ with respect to the inverse index order $n,n-1,\ldots,1$ (see Fig.~\ref{fig:backtree}, where $F'$ consists of two trees with roots at the upper endpoints of $l_8^+$ and $l_6^+$, respectively). The path from each leaf $p_i$ to its root in $F'$ is a $y$-monotone leftward convex chain

We remark that the forward/backward forest is very similar to the leftist/rightist skeleton in~\cite{ref:AlevizosAn90} as well as the upper/lower horizon tree in~\cite{ref:EdelsbrunnerTo89}. Note that it might also be possible to use the topologically sweeping method of~\cite{ref:EdelsbrunnerTo89} to construct $Z^+(\ell)$. However, the algorithm needs to dynamically update the upper/lower horizon trees after processing each event. In contrast, as will be seen later, our algorithm does not need to update the forests, which is not only simpler but also simple.

\paragraph{Forest decomposition and a Graham's scan style algorithm.}
We will use both forests $F$ and $F'$ to construct the cells $C^+_i$ for all $1\leq i\leq n$. To this end, we describe a Graham's scan style algorithm to compute $F$ (the other forest $F'$ can be computed analogously). As will be seen, the algorithm also naturally decomposes $F$ into $n$ $y$-monotone rightward convex chains $\alpha_1,\alpha_2,\ldots,\alpha_n$, such that $F$ is the edge-disjoint union of all these chains (i.e., each edge of $F$ belongs to one and only one chain; see Fig.~\ref{fig:decomp}).
For each $i$, the lower endpoint of $\alpha_i$ is $p_i$, and $\alpha_i$ is a sub-path of the path from $p_i$ to the root of the tree containing $p_i$. Specifically, $\alpha_n$ is the path from $p_n$ to its root in $F$. For each $1\leq i\leq n-1$, $\alpha_i$ is the sub-path from $p_i$ to its tree root until the first vertex in the union of the chains $\alpha_{i+1},\alpha_{i+2},\ldots,\alpha_n$.

\begin{figure}[t]
\begin{minipage}[t]{\textwidth}
\begin{center}
\includegraphics[height=2.5in]{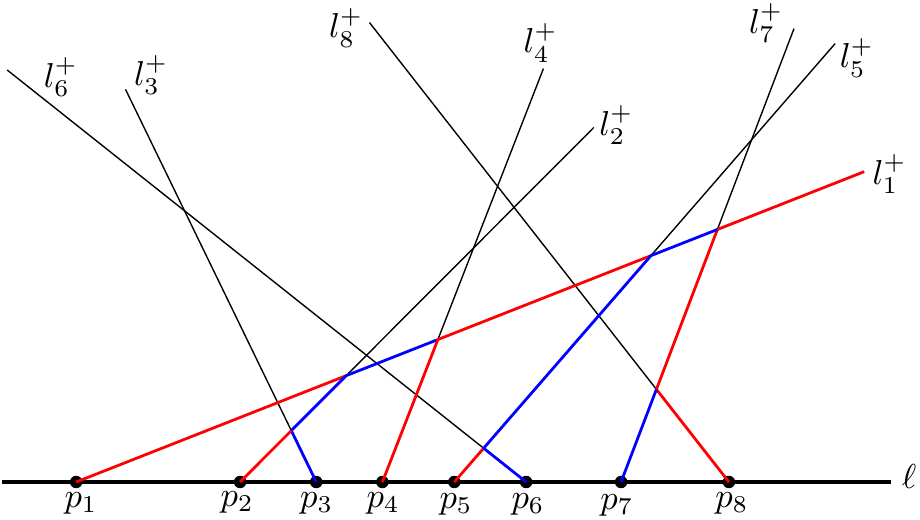}
\caption{\footnotesize Illustrating the decomposition of $F$ into $y$-monotone rightward convex chains $\alpha_i$, $1\leq i\leq n$: Each chain $\alpha_i$ is distinguished with the same color starting from the point $p_i$. For example, the blue chain starting from $p_6$ is $\alpha_6$, which has three segments.}
\label{fig:decomp}
\end{center}
\end{minipage}
\vspace{-0.15in}
\end{figure}

Our algorithm will maintain a $y$-monotone rightward convex chain, denoted by $\alpha$.
Initially, we set $\alpha=l_1^+$. We process $l_i^+$ iteratively following the order $i=2,3,\ldots,n$. For $i=2$, we first check whether $l_2^+$ intersects $\alpha$. If yes (let $q$ be the intersection of $l_2^+$ and $\alpha$), then $\alpha_1$ is defined as $\overline{p_1q}$ and $\alpha$ is updated to the union of $\overline{p_2q}$ and $\overline{qp_1'}$ (recall that $p_1'$ denotes the upper endpoint of $l_1^+$ at infinity); otherwise, $\alpha_1$ is defined as $l_1^+$ and $\alpha$ is updated to $l_2^+$. In general, right before the $i$-th iteration (i.e., right before $l_i^+$ is processed), a $y$-monotone rightward convex chain $\alpha$ is maintained and the lower endpoint of $\alpha$ is $p_{i-1}$ (see Fig.~\ref{fig:chain}). The general algorithm for the $i$-th iteration works as follows. By traversing on $\alpha$ from its lower endpoint $p_{i-1}$, we find the intersection $q$ between $\alpha$ and $l_i^+$ (see Fig.~\ref{fig:chain}). If $q$ exists, then $\alpha_{i-1}$ is defined as the portion of $\alpha$ between $p_{i-1}$ and $q$, and $\alpha$ is updated to the union of $\overline{p_iq}$ and $\alpha\setminus \alpha_{i-1}$. If $q$ does not exist, then $\alpha_{i-1}=\alpha$ and $\alpha=l_i^+$. In either case, the new $\alpha$ is  a $y$-monotone rightward convex chain with lower endpoint at $p_i$. We then proceed on the next iteration. The algorithm stops once $l^+_n$ is processed, after which $F$ and all chains $\alpha_1,\alpha_2,\ldots,\alpha_n$ are obtained.

\begin{figure}[h]
\begin{minipage}[t]{\textwidth}
\begin{center}
\includegraphics[height=2.5in]{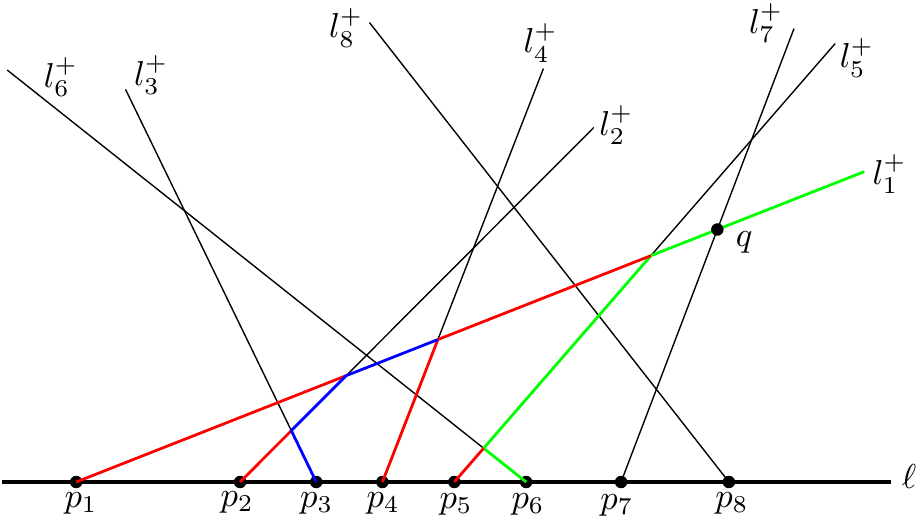}
\caption{\footnotesize Illustrating the chain $\alpha$ (colored green) right before $l_7^+$ is processed. After $l_7^+$ is processed, $\alpha_6$ becomes the portion of $\alpha$ between $p_6$ and $q$, and $\alpha$ is updated to $\overline{p_7q}\cup \overline{qp_1'}$, where $p_1'$ is the upper endpoint of $l_1^+$ at infinity.}
\label{fig:chain}
\end{center}
\end{minipage}
\vspace{-0.15in}
\end{figure}

For the time analysis, observe that each $i$-th iteration takes time proportional to the number of edges of the chain $\alpha_{i-1}$ because we traverse $\alpha$ starting from $p_{i-1}$. As $F$ is the edge-disjoint union of all chains $\alpha_{i}$, $1\leq i\leq n$, the total number of edges of all chains is equal to the number of edges of $F$, which is at most $2n-1$. Hence, the time of the algorithm is $O(n)$.


\begin{figure}[t]
\begin{minipage}[t]{\textwidth}
\begin{center}
\includegraphics[height=2.5in]{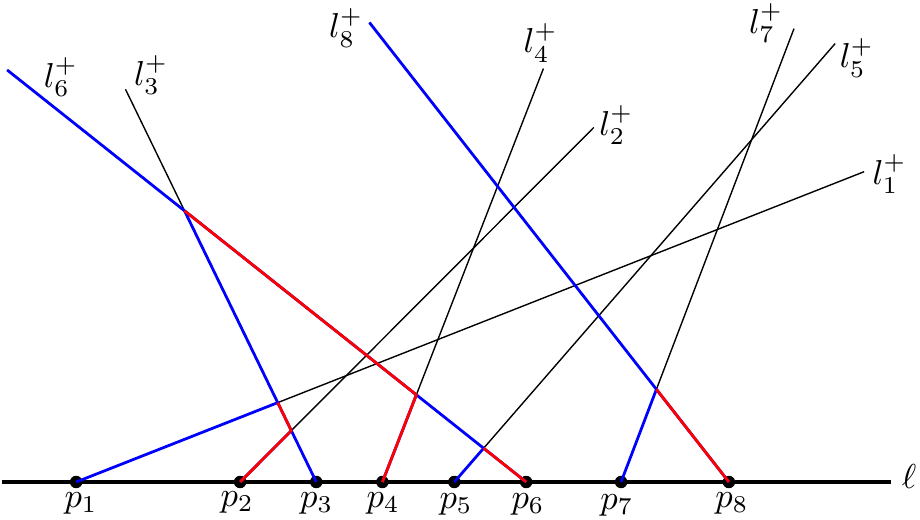}
\caption{\footnotesize Illustrating the decomposition of $F'$ into $y$-monotone leftward convex chains $\beta_i$, $1\leq i\leq n$: Each chain $\beta_i$ is distinguished with the same color starting from the point $p_i$. For example, the blue chain starting from $p_7$ is $\beta_7$, which has two segments.}
\label{fig:decompback}
\end{center}
\end{minipage}
\end{figure}

Analogously, we can decompose the backward forest $F'$ into $n$ $y$-monotone leftward convex chains $\beta_1, \beta_2, \ldots, \beta_n$, such that $F'$ is the edge-disjoint union of all these chains (see Fig.~\ref{fig:decompback}). For each $i$, the lower endpoint of $\beta_i$ is $p_i$. Specifically, $\beta_1$ is the path from $p_1$ to its root in $F'$. For each $2\leq i\leq n$, $\beta_i$ is the sub-path from $p_i$ to its tree root until the first vertex in the union of the chains $\beta_{1},\beta_{2},\ldots,\beta_{i-1}$.
The forest $F'$, along with the chains $\beta_i$, $1\leq i\leq n$, can be computed in $O(n)$ time by an algorithm similar to the above for $F$.


\paragraph{Computing the cells $C_i^+$.}
We are now ready to compute the cells $C_i^+$, $1\leq i\leq n$, using the convex chains of $F$ and $F'$ computed above. To this end, the following lemma is critical.
Let $\partial C_i^+$ denote the boundary of $C_i^+$.

\begin{figure}[t]
\begin{minipage}[t]{\textwidth}
\begin{center}
\includegraphics[height=2.5in]{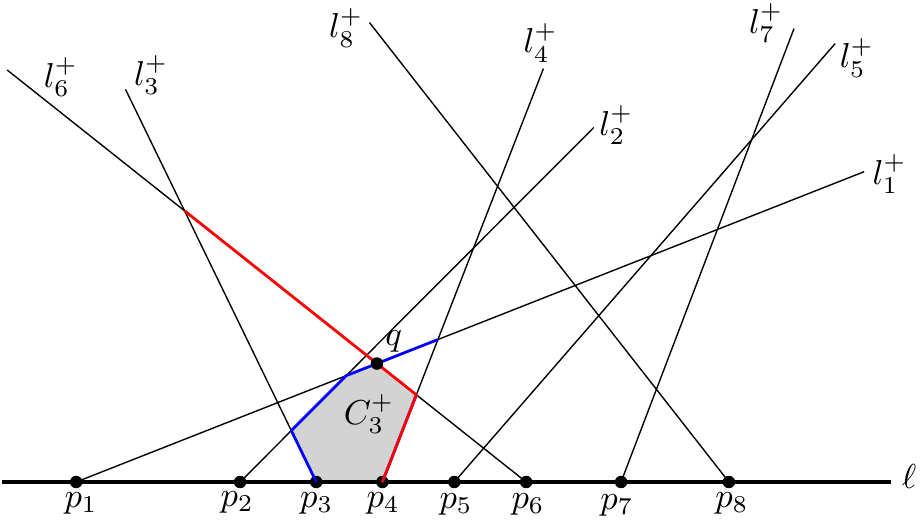}
\caption{\footnotesize The blue chain is $\alpha_3$ and the red chain is $\beta_4$, and they intersect at $q$. The grey region is $C_3^+$, whose boundary consists of $\overline{p_3p_{4}}$, the portion of $\alpha_3$ between $p_3$ and $q$, and the portion of $\beta_{4}$ between $p_{4}$ and $q$, with $q=\alpha_3\cap \beta_4$.}
\label{fig:cellbound}
\end{center}
\end{minipage}
\vspace{-0.15in}
\end{figure}

\begin{lemma}\label{lem:10}
For each $1\leq i\leq n$, the boundary $\partial C_i^+$ can be identified as follows.
\begin{itemize}
  \item For $1\leq i\leq n-1$, if $\alpha_i$ and $\beta_{i+1}$ intersect, say, at a point $q$, then $\partial C_i^+$ consists of the following three parts: $\overline{p_ip_{i+1}}$, the portion of $\alpha_i$ between $p_i$ and $q$, and the portion of $\beta_{i+1}$ between $p_{i+1}$ and $q$ (see Fig.~\ref{fig:cellbound}); otherwise, $\partial C_i^+=\overline{p_ip_{i+1}}\cup \alpha_i\cup\beta_{i+1}$.
  \item For $i=n$, $\partial C_i^+=\overline{p_np_{n+1}}\cup \alpha_n$ (see Fig.~\ref{fig:decomp}).
  \item For $i=0$, $\partial C_i^+=\overline{p_0p_{1}}\cup \beta_1$ (see Fig.~\ref{fig:decompback}).
\end{itemize}
\end{lemma}
\begin{proof}
We only prove the general case $1\leq i\leq n-1$, since the other two special cases can be proved analogously (and in an easier way).

In the following, unless otherwise stated, all points in question are above $\ell$. Also, when we say a point $p$ is to the right (resp., left) of a half-line $l_i^+$, it includes the case $p\in l_i^+$, i.e., the $x$-coordinate of $p$ is larger than (resp., smaller than) or equal to that of $p'$, where $p'$ is the intersection between $l_i^+$ and the horizontal line through $p$.

We first have the following observation about $C_i^+$: a point $p$ is in $C^+_i$ if and only if $p$ is to the right of $l^+_j$ for all $j\leq i$ and is also to the left of $l^+_j$ for all $j\geq i+1$.

\begin{figure}[h]
\begin{minipage}[t]{\textwidth}
\begin{center}
\includegraphics[height=2.5in]{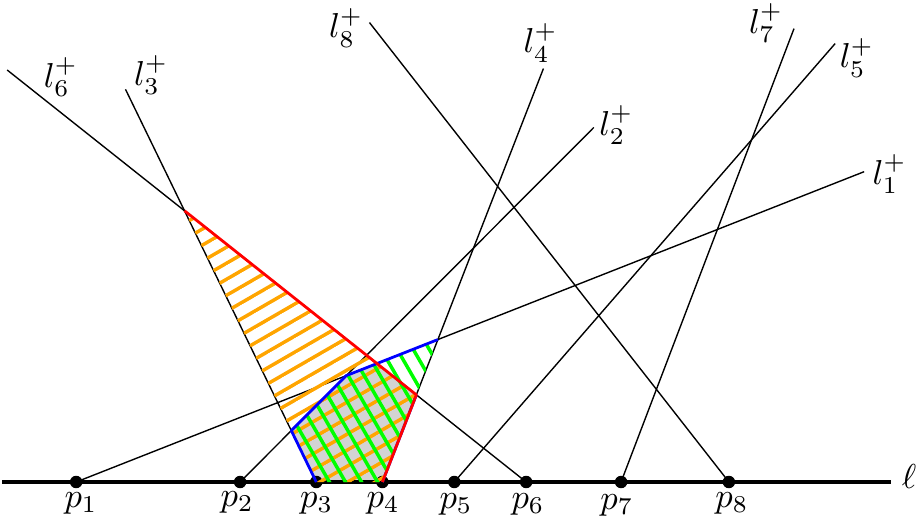}
\caption{\footnotesize Illustrating the regions $A_3$ (green backslash region) and $B_4$ (orange slash region). The gray region, which is $A_3\cap B_4$, is $C_3^+$. The blue chain is $\alpha_3$ and the red chain is $\beta_4$.}
\label{fig:regions}
\end{center}
\end{minipage}
\end{figure}

Recall that both $\alpha_i$ and $\beta_{i+1}$ are $y$-monotone convex chains with their lower endpoints at $\ell$. Let $A_i$ refer to the region of the plane bounded by $\overline{p_ip_{i+1}}$, $\alpha_i$, and the segment of $l_{i+1}^+$ on $F$ (see Fig.~\ref{fig:regions}). Note that if $l_{i+1}^+$ does not intersect $\alpha_i$, then the upper endpoint of $\alpha_i$ is at infinity and the entire $l_{i+1}^+$ is on $F$, and thus $A_i$ is unbounded. Since $\alpha_i$ is $y$-monotone, a point $p$ is in $A_i$ if and only if $p$ is to the left of $l_{i+1}^+$ and also to the right of $\alpha_i$ (i.e., the horizontal line through $p$ intersects $\alpha_i$ and $p$ is to the right of the intersection). Further, by the definition of $\alpha_i$, $A_i$ consists of all points to the left of $l_{i+1}^+$ and to the right of the right envelope of the lines $\{l_1^+,l_2^+,\ldots,l_i^+\}$. Therefore, by the property of the right envelope, we obtain that a point $p$ is in $A_i$ if and only if $p$ is to the left of $l_{i+1}^+$ and to the right of $l_j^+$ for all $j\leq i$. Due to the above observation on $C^+_i$, we obtain that $C^+_i\subseteq A_i$.

Similarly, define $B_{i+1}$ as the region of the plane bounded by $\overline{p_ip_{i+1}}$, $\beta_{i+1}$, and the segment of $l_{i}^+$ on $F'$ (see Fig.~\ref{fig:regions}). By an argument similar to the above, we have $C^+_i\subseteq B_{i+1}$.
Therefore, it follows that $C^+_i=A_i\cap B_{i+1}$ (see Fig.~\ref{fig:regions}).

Recall that the boundary $\partial A_i$ consists of $\overline{p_ip_{i+1}}$, $\alpha_i$, and the segment of $l_{i+1}^+$ on $F$. We call $\alpha_i$ and the segment of $l_{i+1}^+$ on $F$ the {\em left and right boundaries} of $A_i$, respectively. Similarly, we call $\beta_{i+1}$ and the segment of $l_{i}^+$ on $F'$ the {\em right and left boundaries} of $B_{i+1}$, respectively. For $C^+_i$, recall that it is convex and has $\overline{p_ip_{i+1}}$ as an edge. We define its left and right boundaries similarly. Specifically, if $C^+_i$ is bounded, then the highest vertex of $C_i^+$ partitions $\partial C_i^+\setminus \overline{p_ip_{i+1}}$ into two $y$-monotone convex chains, and the left one is called the {\em left boundary} of $C_i^+$ and the right one is called the {\em right boundary}. If $C_i^+$ is unbounded, then $\partial C_i\setminus \overline{p_ip_{i+1}}$ consists of two $y$-monotone convex chains going upwards to infinity, and the left one is called the {\em left boundary} of $C_i^+$ and the right one is called the {\em right boundary}. In either case, the left boundary of $C_i^+$ is $y$-monotone rightward convex and the right boundary of $C_i^+$ is $y$-monotone leftward convex.

Because $C_i^+=A_i\cap B_{i+1}$, and $A_i$, $B_{i+1}$, and $C_i^+$ are all bounded by $\overline{p_ip_{i+1}}$ from below, the left boundary of $C_i^+$ belongs to the right envelope of the left boundary of $A_i$ and the left boundary of $B_{i+1}$. Recall that if we move on the left boundary of $A_i$ from $p_i$, then we first move on $l_i^+$ and then always turn rightwards. On the other hand, the left boundary of $B_{i+1}$ is a segment of $l_i^+$ with $p_i$ as an endpoint. This implies that the left boundary of $C_i^+$ must be a subset of the left boundary of $A_i$, i.e., $\alpha_i$, because the left boundaries of $A_i$, $B_{i+1}$, and $C_i^+$ are all rightward convex. A symmetric argument shows that the right boundary of $C_i^+$ must be a subset of the right boundary of $B_{i+1}$, i.e., $\beta_{i+1}$. Therefore, if $C_i^+$ is closed, then the left boundary of $\partial C_i^+$ is the portion of $\alpha_i$ between $p_i$ and $q$, and the right boundary of $\partial C_i^+$ is the portion of $\beta_{i+1}$ between $p_{i+1}$ and $q$, where $q$ is the intersection of $\alpha_i$ and $\beta_{i+1}$. If $C_i^+$ is open, then the left boundary of $C_i^+$ is $\alpha_i$ and the right boundary of $C_i^+$ is $\beta_{i+1}$. This proves the lemma.
\end{proof}

Based on Lemma~\ref{lem:10}, the cells $C_i^+$, $1\leq i\leq n$, can be easily computed.
First of all, the two cells $C_0^+$ and $C_n^+$ are already available because both $\alpha_n$ and $\beta_1$ have been computed. For each $1\leq i\leq n-1$, we can compute $C_i^+$ using the two convex chains $\alpha_i$ and $\beta_{i+1}$ as follows. By Lemma~\ref{lem:10}, it suffices to find the intersection $q$ between $\alpha_i$ and $\beta_{i+1}$, or determine that such an intersection does not exist. This can be done in $O(|\alpha_i|+|\beta_{i+1}|)$ time by a straightforward sweeping algorithm similar to that for merging two sorted lists. Indeed, starting from $\ell$, we sweep a horizontal line $h$ upwards. During the sweeping, we maintain the edges of $\alpha_i$ and $\beta_{i+1}$ intersecting $h$. An event happens if $h$ encounters a vertex of $\alpha_i$ or $\beta_{i+1}$. Consider an event at a vertex $a\in \alpha_i$, i.e., $a\in h$. Let $a'$ be the next vertex of $\alpha_i$ after $a$, i.e., $\overline{aa'}$ is the edge of $\alpha_i$ right above $h$ (see Fig.~\ref{fig:sweep}). Let $\overline{bb'}$ be the edge of $\beta_{i+1}$ currently intersecting $h$ such that $b'$ is the upper vertex of the edge. To process the event, we first check whether $\overline{aa'}$ and $\overline{bb'}$ intersect. If yes, then their intersection is $q$ and we stop the algorithm. Otherwise, we proceed on the next event, which is the lower point of $a'$ and $b'$. If $q$ is not found after all events are processed, then $\alpha_i$ and $\beta_{i+1}$ do not intersect and we stop the algorithm. It is easy to see that the algorithm runs in $O(|\alpha_i|+|\beta_{i+1}|)$ time.

\begin{figure}[t]
\begin{minipage}[t]{\textwidth}
\begin{center}
\includegraphics[height=1.5in]{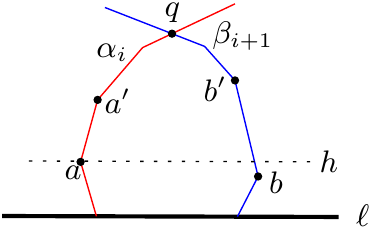}
\caption{\footnotesize Processing an event at $a$.}
\label{fig:sweep}
\end{center}
\end{minipage}
\vspace{-0.15in}
\end{figure}

Therefore, computing all cells $C_i^+$ for all $i=0,1,\ldots,n$ takes time proportional to $\sum_{i=1}^{n}|\alpha_i|+\sum_{i=1}^{n}|\beta_i|$, which is $O(n)$. Consequently, $Z^+(\ell)$, which is the disjoint union of all cells $C_i^+$, for all $i=0,1,\ldots,n$, is obtained. Note that since we already have the sorted list $p_1,p_2,\ldots,p_n$, we could also easily connect cells $C_i^+$ together from left to right on $\ell$.

Again, we can use a similar algorithm to compute $Z^-(\ell)$ in $O(n)$ time. Finally, the zone $Z(\ell)$ is simply the disjoint union of $Z^+(\ell)$ and $Z^-(\ell)$. The following theorem summarizes our result.

\begin{theorem}
Given a set $L$ of $n$ lines in the plane and another line $\ell$, the zone $Z(\ell)$ in the arrangement of $L$ can be computed in $O(n\log n)$ time. If the sorted list of the intersections between $\ell$ and all lines of $L$ is known, then $Z(\ell)$ can be computed in $O(n)$ time.
\end{theorem}

\paragraph{Dealing with degeneracies.}
There are two degenerate cases: (1) $L$ has horizontal lines; (2) $\ell$ contains intersections of lines of $L$.

To handle the first case, without loss of generality, we assume that $L$ has horizontal lines above $\ell$, and among those, let $\ell'$ be the lowest one. We first compute the ``upper zone'' $Z^+(\ell)$ as usual without considering the horizontal lines of $L$. Then, for each cell $C_i^+$ of $Z^+(\ell)$, we simply cut it along $\ell'$ (alternatively, we could easily incorporate this cut operation into our sweeping algorithm for computing $C_i^+$). The union of all cells $C_i^+$ after the cut is the upper zone $Z^+(\ell)$ of $L$ including all horizontal lines.
The total time of the algorithm does not change asymptotically.

To handle the second case, what really matters is the sorted list $l_1,l_2,\ldots,$ of $L$, which is used in our forest decomposition algorithm. If two or more lines of $L$ have a common intersection on $\ell$, then we break the tie by further comparing their slopes: a line $l$ is placed in the sorted list in the front of another line $l'$ if the half-line of $l$ above $\ell$ is left of that of $l'$.
After having the sorted list of $L$, we can run exactly the same algorithm as before.

\paragraph{Remark.} Our algorithm also provides a simple proof for the combinatorial size of the zone $Z(\ell)$.
If $L$ has a horizontal line, a slight rotation of it only increases the complexity of $Z(\ell)$. Hence, it suffices to assume that $L$ does not have any horizontal line.
According to our algorithm, each edge of $Z^+(\ell)$ lies on an edge of one of the two forests $F$ and $F'$. As discussed before, each forest has at most $2n-1$ edges (including the degenerate cases). Hence, the total number of edges of $Z^+(\ell)$ is at most $4n-2$. Therefore, the total number of edges of the zone $Z(\ell)$ is at most $8n-4$, which matches the bounds obtained in~\cite{ref:ChazelleTh85,ref:EdelsbrunnerAr92,ref:EdelsbrunnerCo86}.




\begin{thebibliography}{10}

\bibitem{ref:AgarwalPa902}
P.K. Agarwal.
\newblock Partitioning arrangements of lines {II}: Applications.
\newblock {\em Discrete and Computational Geometry}, 5:533--573, 1990.

\bibitem{ref:AgarwalCo98}
P.K. Agarwal, J.~Matou\v{s}ek, and O.~Schwarzkopf.
\newblock Computing many faces in arrangements of lines and segments.
\newblock {\em SIAM Journal on Computing}, 27:491--505, 1998.

\bibitem{ref:AgarwalAr00}
P.K. Agarwal and M.~Sharir.
\newblock {\em {\em Arrangements and their applications, in} Handbook of
  Computational Geometry, {\em J. Sack and J. Urrutia (eds.)}}, pages 49--119.
\newblock Elsevier, Amsterdam, The Netherlands, 2000.

\bibitem{ref:AlevizosAn90}
P.~Alevizos, J.-D. Boissonnat, and F.P. Preparata.
\newblock An optimal algorithm for the boundary of a cell in a union of rays.
\newblock {\em Algorithmica}, 5:573--590, 1990.

\bibitem{ref:BernHo91}
M.W. Bern, D.~Eppstein, P.E. Plassmann, and F.F. Yao.
\newblock Horizon theorems for lines and polygons.
\newblock {\em Discrete and Computational Geometry: Papers from the {DIMACS}
  Special Year}, 6:45--66, 1991.

\bibitem{ref:ChazelleTh85}
B.~Chazelle, L.J. Guibas, and D.T. Lee.
\newblock The power of geometric duality.
\newblock {\em BIT}, 25:76--90, 1985.

\bibitem{ref:EdelsbrunnerAl87}
H.~Edelsbrunner.
\newblock {\em Algorithmis in Combinatorial Geometry}.
\newblock Springer-Verlag, Heidelberg, Germany, 1987.

\bibitem{ref:EdelsbrunnerTo89}
H.~Edelsbrunner and L.~Guibas.
\newblock Topologically sweeping an arrangement.
\newblock {\em Journal of Computer and System Sciences}, 38(1):165--194, 1989.

\bibitem{ref:EdelsbrunnerIm89}
H.~Edelsbrunner, L.~Guibas, J.~Hershberger, R.~Seidel, M.~Sharir, J.~Snoeyink,
  and E.~Welzl.
\newblock Implicitly representing arrangements of lines or segments.
\newblock {\em Discrete and Computational Geometry}, 4:433--466, 1989.

\bibitem{ref:EdelsbrunnerAr92}
H.~Edelsbrunner, L.~Guibas, J.~Pach, R.~Pollack, R.~Seidel, and M.~Sharir.
\newblock Arrangements of curves in the plane topology, combinatorics, and
  algorithms.
\newblock {\em Theoretical Computer Science}, 92(2):319--336, 1992.

\bibitem{ref:EdelsbrunnerSi90}
H.~Edelsbrunner and E.P. M{\"u}cke.
\newblock Simulation of simplicity: A technique to cope with degenerate cases
  in geometric algorithms.
\newblock {\em ACM Transactions on Graphics}, 9:66--104, 1990.

\bibitem{ref:EdelsbrunnerCo86}
H.~Edelsbrunner, J.~O'Rourke, and R.~Seidel.
\newblock Constructing arrangements of lines and hyperplanes with applications.
\newblock {\em SIAM Journal on Computing}, 15:341--363, 1986.

\bibitem{ref:EdelsbrunnerOn93}
H.~Edelsbrunner, R.~Seidel, and M.~Sharir.
\newblock On the zone theorem for hyperplane arrangements.
\newblock {\em SIAM Journal on Computing}, 22:418--429, 1993.

\bibitem{ref:HalperinAr17}
D.~Halperin and M.~Sharir.
\newblock {\em {\em Arrangements, in} Handbook of Discrete and Computational
  Geometry, {\em C.D. T\'{o}th, J. O'Rourke, and J.E. Goodman (eds.)}}, pages
  723--762.
\newblock CRC Press, 3rd edition, 2017.

\bibitem{ref:SharirDa95}
M.~Sharir and P.K. Agarwal.
\newblock {\em Davenport-Schinzel Sequences and Their Geometric Applications}.
\newblock Cambridge University Press, 1995.

\end{thebibliography}

\end{document}